\documentclass[11pt]{article}

\newcommand{\Anote}[1]{}
\newcommand{\Pnote}[1]{}

\usepackage{geometry}
\usepackage{setspace}
\usepackage{enumerate}

\usepackage{graphicx}
\usepackage{mathrsfs,amssymb,amsmath,textcomp,amsfonts,amsthm,bbm,fullpage}
\usepackage[backref, colorlinks,citecolor=blue,bookmarks=true]{hyperref}

\usepackage{times}

\newif\ifhyper\IfFileExists{hyperref.sty}{\hypertrue}{\hyperfalse}
\hypertrue
\ifhyper\usepackage{hyperref}\fi

\usepackage{color}
\usepackage{multirow,amsmath, dsfont, amscd, latexsym,color,amssymb,setspace}
\usepackage{ifthen}
\usepackage[noend]{algorithmic}
\usepackage{enumerate}

\newtheorem{Thm}{Theorem}
\newtheorem{Lem}[Thm]{Lemma}
\newtheorem{Cor}[Thm]{Corollary}

\newtheorem{Claim}[Thm]{Claim}
\newtheorem{Conj}[Thm]{Conjecture}
\newtheorem{Def}[Thm]{Definition}
\newtheorem{Prob}[Thm]{Problem}

\def\bull{\vrule height .9ex width .8ex depth -.1ex }

\newenvironment{Proof}{\medbreak
\noindent {\bf Proof:~}}{\unskip\nobreak\hfill\hskip 2em \bull\par\medbreak}

\newcommand{\eat}[1]{}

\newcommand{\rgta}{\ensuremath{\rightarrow}}

\newcommand{\fr}[1]{\frac{1}{#1}}

\DeclareMathOperator{\wt}{wt}

\DeclareMathOperator{\s}{s}

\newcommand{\noise}{\mathrm{Noise}}
\newcommand{\Bin}{\mathrm{Bin}}

\newcommand{\ind}{\ensuremath \mathbbm{1}}
\newcommand{\trf}{\ensuremath T_{1 - 2\delta} f}
\newcommand{\nsf}{\ensuremath \mathrm{NS}_\delta[f]}

\newcommand{\GS}{\ensuremath \mathrm{\Lambda}_{\delta,\theta}(S)}
\newcommand{\tr}{\ensuremath \tilde{r}}
\newcommand{\zo}{\{0,1\}}
\newcommand{\zon}{\{0,1\}^n}

\newcommand{\commented}{yes}

\ifthenelse{\equal{\commented}{yes}}{
\newcommand{\pnote}[1]{\footnote{{\bf [[Parik: {#1}\bf ]] }}}
\newcommand{\rnote}[1]{\footnote{{\bf [[Rocco: {#1}\bf ]] }}}
\newcommand{\anote}[1]{\footnote{{\bf [[Avi: {#1}\bf ]] }}}

}

\ifthenelse{\equal{\commented}{no}}{
\newcommand{\pnote}[1]{}
\newcommand{\rnote}[1]{}
\newcommand{\anote}[1]{}
}

\newcommand{\ignore}[1]{}

\newcommand{\blue}[1]{{\color{blue} #1}}

\newcommand{\Maj}[2]{\ensuremath{\mathop{\mathrm{Maj}}_{#1}(#2)}}
\newcommand{\rl}{colex}

\newcommand{\E}{{\bf E}}
\newcommand{\Ex}{\mathop{\mathbb{E}\/}}

\newcommand{\R}{\mathbb R}

\newcommand{\F}{\mathbb F}
\newcommand{\Z}{\mathbb Z}

\newcommand{\eps}{\epsilon}

\newcommand{\poly}{\mathrm{poly}}

\newcommand{\supp}{\mathrm{supp}}

\newcommand{\sbkets}[1]{\ensuremath \left[ #1 \right]}
\newcommand{\cbkets}[1]{\ensuremath \left\{ #1 \right\}}
\newcommand{\bkets}[1]{\ensuremath \left( #1 \right)}

\newcommand{\ms}[1]{\ensuremath \mathsf{#1}}

\newcommand{\calF}{{\cal F}}

\newcommand{\calS}{{\cal S}}
\newcommand{\calT}{{\cal T}}

\newcommand{\calB}{{\cal B}}

\newcommand{\cS}{\calS}
\newcommand{\cB}{\calB}
\newcommand{\cT}{\calT}

\newcommand{\Mj}{\mathrm{Maj}}
\newcommand{\Par}{\mathrm{Par}}

\newcommand{\myalgo}[2]{
\bigskip
\noindent \fbox{
\begin{minipage}{6in}
{\bf #1}\\
{\tt #2}
\end{minipage}}
\bigskip
}

\newcounter{this-list}

\begin{document}



\title{Smooth Boolean functions are easy: \\ efficient algorithms for low-sensitivity functions}

\author{
Parikshit Gopalan\\
Microsoft Research\\
{\tt parik@microsoft.com} 
\and 
Noam Nisan\\ 
Microsoft Research \& \\ 
The Hebrew University\\
{\tt noam.nisan@gmail.com}
\and 
Rocco A. Servedio\\
Columbia University \\
{\tt rocco@cs.columbia.edu}
\and
Kunal Talwar\\
Google Research\\
{\tt ktalwar@gmail.com}\\
\and
Avi Wigderson\\
IAS\\
{\tt avi@ias.edu}
}

\begin{titlepage}
\maketitle

\begin{abstract}

A natural measure of smoothness of a  Boolean function is  its {\em sensitivity} (the largest number of Hamming neighbors of a point which differ from it in function value). The structure of smooth or equivalently low-sensitivity functions is still a mystery.
A well-known conjecture states that every such Boolean function can
be computed by a shallow decision tree. While this conjecture implies
that smooth functions are easy to compute in the {\em simplest}
computational model, to date no non-trivial upper bounds were known
for such functions in {\em any} computational model,
including unrestricted Boolean circuits.  Even a bound on the
description length of such functions better than the trivial $2^n$
does not seem to have been known. 

In this work, we establish the first computational upper bounds on
smooth Boolean functions:  

\begin{itemize}
\item We show that every sensitivity $s$ function is uniquely
  specified by its values on a Hamming ball of radius $2s$. We use this
  to show that such functions can be computed by circuits of size
  $n^{O(s)}$.

\item We show that sensitivity $s$ functions satisfy a strong {\em pointwise}
  noise-stability guarantee for random noise of rate $O(1/s)$. We
  use this to show that these functions have formulas of depth $O(s\log n)$.

\item We show that sensitivity $s$ functions can be (locally)
  self-corrected from worst-case noise of rate $\exp(-O(s))$. 
\end{itemize}

All our results are simple, and follow rather directly from (variants of)
the basic fact that that the function value at few points in small
neighborhoods of a given point determine its function value via a majority vote. Our
results confirm various consequences of the conjecture. They may be
viewed as providing a new form of evidence towards its validity, as
well as new directions towards attacking it. 
\end{abstract}

 \thispagestyle{empty}

\end{titlepage}

\section{Introduction}

\subsection{Background and motivation}

The smoothness of a continuous function captures how gradually it
changes locally (according to the metric of the underlying space). For
Boolean functions on the Hamming cube, a natural analog is {\em
  sensitivity}, capturing how many neighbors of a point have different
function values. More formally, the \emph{sensitivity} of a Boolean
function $f:\zon \rgta \zo$ at input $x \in \zon$, written $s(f,x)$,
is   the number of neighbors $y$ of $x$ in the Hamming cube such that $f(y) \neq f(x).$
The \emph{max sensitivity} of $f$, written $s(f)$ and often referred to simply as the ``sensitivity of $f$'', is defined as
\ignore{
\begin{align*}
s(f) & =\max_{x\in \zon}s(f,x).
\end{align*}
}
$s(f)  = \max_{x\in \zon}s(f,x).$
So, $0 \leq s(f) \leq n$, and while not crucial, it may be good for
the reader to consider this parameter as ``low'' when e.g. either
$s(f) \leq (\log n)^{O(1)}$ or $s(f)\leq n^{o(1)}$ (note that both
upper bounds are closed under taking polynomials). 

\Pnote{We should perhaps say more about the analogy to the continuous
  setting? Added:}  
To see why low-sensitivity functions might be considered smooth, let $\delta(\cdot,\cdot)$
denote the normalized Hamming metric on $\zo^n$. A simple
application of the triangle inequality gives
\[\E_{y: \delta(x,y) = \delta_0}|f(x) - f(y)| \leq \delta_0 s(f).\]
Thus $s(f)$ might be viewed as being somewhat analogous to the Lipschitz
constant of $f$. 

\Pnote{Another concern is that the term smoothness disappears after Section 1.1. Do any heroes volunteer to insert it where appropriate?}

A well known conjecture states that every smooth Boolean function is
computed by a shallow decision tree, specifically of depth polynomial in the sensitivity.  This conjecture was first posed  in the form of  a question by Nisan \cite{Nisan:91} and Nisan and Szegedy  \cite{NisanSzegedy:94} but is now (we feel) widely believed to be true:

\begin{Conj}\cite{Nisan:91, NisanSzegedy:94}
\label{conj:1}
There exists a constant $c$ such that every Boolean function $f$ has a decision tree of depth $s(f)^c$. 
\end{Conj}
The converse is trivial, since every Boolean function computable by a depth $d$
decision tree has sensitivity at most $d$. However, the best known upper
bound on decision tree depth in terms of sensitivity is exponential (see Section \ref{sec:related-work}).

A remarkable series of developments, starting with Nisan's paper~\cite{Nisan:91},
showed that decision tree depth is an extremely robust complexity
parameter, in being polynomially related to many other, quite diverse
complexity measures for Boolean functions, including PRAM complexity,
block sensitivity, certificate complexity, randomized decision tree
depth, quantum decision tree depth, real polynomial degree, and
approximating polynomial degree. Arguably the one natural complexity
measure that has defied inclusion in this equivalence class is
sensitivity.  Thus, there are many equivalent
formulations of Conjecture \ref{conj:1}; indeed, Nisan originally posed the
question in terms of sensitivity versus block sensitivity \cite{Nisan:91}.  See the
extensive survey  \cite{Chicago} for much more information about the
conjecture and \cite{BuhrmandeWolf:02} for background on various
Boolean function complexity measures.  

\Pnote{Not sure I see where this paragraph adds, conditioned on the
  last one}
 \Anote{OK, removed, and added a bit later after Conj 2}

Conjecture \ref{conj:1} is typically viewed as a \emph{combinatorial}
statement about the Boolean hypercube.  However, the conjecture also makes a strong assertion about
\emph{computation}, stating that smooth functions have very
low complexity; indeed, the conjecture posits that they are
easy to compute in arguably the simplest computational model ---
deterministic decision trees. This implies that smooth functions easy for many
other ``low-level'' computational models via the following chain of inclusions:
\begin{align*}
\ms{DecTree}\text{-}\ms{depth}(\poly(s)) & \subseteq \mathsf{DNF}\text{-}\ms{width}(\poly(s))  \subseteq \mathsf{AC_0}\text{-}\ms{size}(n^{\poly(s)})\\
& \subseteq \mathsf{Formula}\text{-}\ms{depth}(\poly(s)\log(n)) \subseteq
\mathsf{Circuit}\text{-}\ms{size}(n^{\poly(s)}).
\end{align*}

Given these inclusions, and the widespread interest that Conjecture \ref{conj:1} has attracted in 
the study of Boolean functions, it is perhaps surprising that no non-trivial upper bounds were previously known on low 
sensitivity functions in \emph{any}
 computational model, including unrestricted Boolean circuits. Indeed, a pre-requisite
for a family of functions to have small circuits  is an upper
bound on the number of  functions in the family, or equivalently
on the description length of such functions; even such bounds were not previously known for low sensitivity
functions.  This gap in our understanding
of low sensitivity functions helped motivate the present work.

An equivalent formulation of Conjecture \ref{conj:1} is that every sensitivity $s$ function is computed by
a real polynomial whose degree is upper bounded by some polynomial in
$s$. This is equivalent to saying that the \emph{Fourier expansion} of the
function has degree $\poly(s)$:
\begin{Conj}\cite{Nisan:91,NisanSzegedy:94}
\label{conj:2} (Equivalent to Conjecture~\ref{conj:1})
There exist a constant $c$ such that every Boolean function is computed by a real polynomial of degree $s(f)^c$. 
\end{Conj}
\Anote{Added on smoothness}
Given the analogy between sensitivity and the Lipschitz constant,
this form of the conjecture gives a natural discrete analog of continuous approximations of smooth Lipschitz functions by low-degree polynomials, first obtained for univariate case by Weierstrass~\cite{weierstrass1885}, which has had a huge influence on the development of modern analysis. This lead to a large body of work in approximation theory, and we mention here the sharp quantitative version of the theorem~\cite{jackson1930} and its extension to the multivariate case~\cite{NewSha1964}.

This formulation of the conjecture is also interesting because of the rich
structure of low-degree polynomials that low sensitivity
functions are believed to share. For instance, low-degree
real polynomials on the Boolean cube are easy to interpolate from relatively few values (say over a Hamming ball). 
The interpolation procedure can be
made tolerant to noise, and local (these follow from the fact that
low-degree real polynomials also have low degree over $\F_2$). Again,
our understanding of the structure of low sensitivity
functions was insufficient to establish such properties for them prior to this work. 

Finally, to every Boolean function $f$ one can associate the bipartite graph
$G_f$ which has left and right vertex sets $f^{-1}(0)$ and
$f^{-1}(1)$, and which has an edge $(x,y)$ if the Hamming distance $d(x,y) $ is $1$ and $f(x)
\neq f(y)$.  A function has max sensitivity $s$ if and only the graph $G_f$ has maximum degree at most $s$.
From this perspective one can view Conjectures~\ref{conj:1} and~\ref{conj:2} as a step towards
understanding the graph-theoretic structure of Boolean functions and
relating it to their computational and analytic structure (as captured
by the Fourier expansion).  In this paper, we propose proving various
implications of the conjecture both as a necessary first step towards
the conjecture, and as a means to better understanding low sensitivity
functions from a computational perspective. 

\subsection{Our Results}

Let $\calF(s,n)$ denote the set of Boolean functions on $n$ variables
such that $s(f) \leq s$. We sometimes refer to this class simply as
``sensitivity $s$ functions'' ($n$ will be implicit).

The starting point for our results is an upper bound
stating that low-sensitivity functions can be interpolated from Hamming
balls. This parallels the fact that a degree $d$ polynomial can be
interpolated from its values on a Hamming ball of radius $d$. 

\begin{Thm}
\label{thm:intro-ball}
Every sensitivity $s$ function on $n$ variables is uniquely specified by its values on
any Hamming ball of radius $2s$ in $\zo^n$.
\end{Thm}

The simple insight here is that knowing the values of $f$ at any set of $2s +1$ neighbors
of a point $x$ uniquely specifies the value of $f$ at $x$: it is the majority
value over the $2s+1$ neighbors (else the point $x$ would be
too sensitive).  This implies the following upper bound on the number
of sensitivity $s$ functions: 
\begin{align*}
|\calF(s,n)| \leq 2^{n \choose \leq 2s}.
\end{align*}

Our proof of Theorem \ref{thm:intro-ball} is algorithmic (but
inefficient). We build on it to give efficient algorithms that
compute $f$ at any point $x \in \zo^n$, given the values of $f$ on a Hamming ball as advice. 

Our first algorithm takes a bottom-up approach. We know the values of $f$ on a
ball of small radius around the origin, and wish to infer
its value at some arbitrary point $x$. Imagine moving the center of
the ball from the origin to $x$ along a shortest path. The key
observation is that after shifting the ball by Hamming distance $1$,
we can recompute the values of $f$ on the shift using a simple
Majority vote.

Our second algorithm uses a top-down approach, reducing computing $f$
at $x$ to computing $f$ at $O(s)$ neighboring points of Hamming weight
one less than $x$. We repeat this till we reach points of weight
$O(s)$ (whose values we know from the advice). By carefully choosing
the set of $O(s)$ neighbors, we  ensure that no more than $n^{O(s)}$
values need to be computed in total:   

\begin{Thm}
\label{thm:ckt}
Every sensitivity $s$ function is computed by a Boolean circuit of size
$O(sn^{2s+1})$ and depth $O(n^s)$.
\end{Thm}
Simon has shown that every sensitivity $s$ function depends on at most
$2^{O(s)}$ variables \cite{Simon82}. Thus, the circuit we construct has size at most
$2^{O(s^2)}$.   

A natural next step would be to parallelize this algorithm. Towards this goal, 
we show that low sensitivity functions satisfy a very strong
noise-stability guarantee: Start at any point $x \in \zo^n$ and take a
random walk of length $n/10s$ to reach a point $y$. Then $f(x) = f(y)$
with probability $0.9$, where the probability is only over the coin
tosses of the walk and not over the starting point $x$. Intutitively,
this says that the value of $f$ at most points in a ball of radius
$n/10s$ around $x$ equals the value at $x$ (note that in contrast, Theorem~\ref{thm:intro-ball} only uses
the fact that most points in a ball of radius $1$ agree with the
center). We use this structural property to get a small depth \emph{formula} that
computes $f$:   

\begin{Thm} 
\label{thm:formulas} 
Every sensitivity $s$ function is computed by a Boolean formula 
of depth $O(s\log n)$ and size $n^{O(s)}.$
\end{Thm}

(By \cite{Simon82}, these formulas have depth at most $O(s^2)$ and size at most $2^{O(s^2)}$ as before.)
At a high level, we again use the the values on a Hamming ball as
advice. Starting from some arbitrary input $x$, we use a variant of the
noise-stability guarantee (which holds for ``downward'' random walks that only flip 1-coordinates to 0) 
to reduce the computation of $f$ at $x$ to computing $f$ on
$O(1)$ many points whose weight is less than that of $x$ by a factor of
roughly $(1- 1/(10s))$ 
(a majority vote on these serves to amplify the success probability). Repeating this for each of these new points, recursively, for $O(s\log(n))$ times, we reduce computing $f$ at $x$ to
computing $f$ at various points in a small Hamming ball around the
origin, which we know from the advice.

We also show that low-sensitivity functions admit 
local self-correction. The setup here is that we are given oracle
access to an unknown function $r: \zo^n \rgta \zo$ that is promised to be close to
a low sensitivity function. Formally, there exists a
sensitivity $s$ function $f:\zo^n \rgta \zo$ such that 
\[ \delta(r,f) := \Pr_{x \in \zo^n}[r(x) \neq f(x)]  \leq 2^{-d s}\]
for some constant $d$.
We are then given an arbitrary $x \in \zo^n$ as an input, and our goal is to return
$f(x)$ correctly with high probability for every $x$, where the probability is over the coin tosses of the
(randomized) algorithm. We show that there is a self-corrector for $f$ with the
following guarantee: 

\Anote{Simplified. Seems to me also that the bound on both queries and
  time is only $(\log(1/\eps))^{O(s)}$ }
\Pnote{I think it is $\log(1/\eps)^{s\log(n/s)}$. So it makes sense to
  first get 0.99 with $(n/s)^{O(s)}$ query complexity, and then one could
  amplify to $1- \eps$ using majority.}

\begin{Thm}
\label{thm:local-intro}
There exist a constant $d$ such that the following holds.
Let  $r:\zo^n \rgta \zo$ be such that $\delta(r,f) \leq 2^{-d s}$
for some sensitivity $s$ function $f$. There is an algorithm which, when
given an oracle for $r$ and $x \in \zo^n$ as input, queries the oracle
for $r$ at $(n/s)^{O(s)}$ points, runs in
$(n/s)^{O(s)}$ time, and returns the correct value of
$f(x)$ with probability $0.99$.  
\end{Thm}  

Our self-corrector is similar in spirit to our formula construction:
our estimate for $f(x)$ is obtained by taking the majority over a
random sample of points in a ball of radius $n/10s$. Rather than
querying these points directly (since they might all be incorrect for an adversarial choice of $x$ and $r$), 
we use recursion. We show that $O(s\log(n))$ levels of recursion guarantee that we compute $f(x)$ with
good probability. The analysis uses Bonami's hypercontractive
inequality \cite{Ryan}. 

Our results imply that low-degree functions and low
sensitivity functions can each be reconstructed from their value on small
Hamming balls using simple but dissimilar looking ``propagation rules''. We show how degree
and sensitivity can be chracterized by the convergence of these respective
propagation rules, and use this to present a reformulaion
of  Conjecture \ref{conj:1}.

\subsection{Related Work} \label{sec:related-work}

The study of sensitivity originated from work on PRAMs
\cite{CookDR:86, Simon82}. As mentioned earlier, the question of relating sensitivity to other
complexity measures such as block sensitivity was posed in \cite{NisanSzegedy:94}.
There has been a large body of work on
Conjecture \ref{conj:1} and its equivalent formulations, and recent years have witnessed significant interest in this problem (see the survey \cite{Chicago} and the papers cited below).  To date, the
biggest gap known between sensitivity and other measures such as
block-sensitivity, degree and decision tree depth is at most
quadratic \cite{Rubinstein, AmbainisS:11}. Upper bounds on other measures such as block sensitivity and certificate complexity in terms of sensitivity are given in \cite{KK, AmbainisBGMSZ:14,AmbainisP:14,AmbainisPV:15} (see also
\cite{AmbainisV:15}).  Very recently, a novel approach to this conjecture via a communication game was proposed in the work of Gilmer {\em et  al.} \cite{GKS:15}.

\subsection{Preliminaries}

We define the $0$-sensitivity, $1$-sensitivity and the max sensitivity
of an $n$-variable function $f$ as
\begin{align*}
s_0(f) =\max_{x\in f^{-1}(0)}s(f,x), \quad
s_1(f) =\max_{x\in f^{-1}(1)}s(f,x), \quad
s(f) =\max_{x\in \zon}s(f,x) = \max(s_0(f),s_1(f)).
\end{align*}
We denote the real polynomial degree of a function by $\deg(f)$ and
its $\F_2$ degree by $\deg_2(f)$.  We write $\wt(x)$ for $x \in
\{0,1\}^n$ to denote the Hamming weight of $x$ (number of ones). 
We write $\delta(f,g)$ for $f,g: \{0,1\}^n \to \{0,1\}$ to denote
$\Pr_{x \in \{0,1\}^n}[f(x) \neq g(x)]$. 

For $x \in \{0,1\}^n$, let $\cB(x,r) \subset \{0,1\}^n$ denote the
Hamming ball consisting of all points at distance at most $r$ from
$x$. Let $\cS(x,r)$ denote the Hamming sphere consisting of all points
at distance exactly $r$ from $x$. Let $N(x)$ denote the set of Hamming
neighbors of $x$ (so $N(x)$ is shorthand for $\cS(x,1)$), and let
$N_r(x)$ denote the set of neighbors of Hamming weight $r$ (points
with exactly $r$ ones).   

The following upper bound on sensitivity in terms of degree is due to Nisan and Szegedy. 
\begin{Thm}\cite{NisanSzegedy:94}
\label{thm:ns}
For every function $f:\zo^n \to \zo$, we have
$\s(f) \leq 4(\deg(f))^2.$
 \end{Thm}

We record Simon's upper bound on the number of relevant variables in a low-sensitivity function:
\begin{Thm}\cite{Simon82}
\label{thm:simon}
For every function $f:\zo^n \to \zo$, the number of relevant variables $n'$ is bounded by
$n' \leq s(f)4^{s(f)}.$
\end{Thm}


\section{Structural properties of low sensitivity functions}

\subsection{Bounding the description length}

We show that functions with low sensitivity have concise descriptions,
so consequently the number of such functions is small. Indeed, we show that knowing
the values on a Hamming ball of radius $2s +1$ suffices.

\subsubsection{Reconstruction from Hamming balls and spheres.}

The following simple but key observation will be used repeatedly:

\begin{Lem}
\label{lem:key}
Let $S \subseteq N(x)$ where $|S| \geq 2s +1$. Then $f(x) = \Maj{y \in S}{f(y)}.$
\end{Lem}
\begin{Proof}
Let $b \in \zo$ denote the majority value of $f$ over $S$ and let
$S^b \subset S$ be the subset of $S$ over which $f$ takes the value $b$. Note that 
$|S^b| \geq \lceil|S|/2\rceil \geq s+1$
since $|S| \geq 2s +1$. If $f(x) \neq b$, then every vertex in $S^b$
represents a sensitive neighbor of $x$, and thus $s(f,x) \geq s+1$
which is a contradiction.
\end{Proof}

\begin{Thm}
\label{thm:ball}
Every sensitivity $s$ function is uniquely specified by its values on
a ball of radius $2s$.
\end{Thm}
\begin{Proof}
Suppose that we know the values of $f$ on $\cB(x,2s)$. We may assume by relabeling that $x =0^n$ is the origin. Note that $\cB(0^n,2s)$ is just the
set of points of Hamming weight at most $2s$. 

We will prove that $f$ is uniquely specified on points $x$ where
$\wt(x) \geq 2s$ by induction on $r =\wt(x)$. The base case $r =2s$ is
trivial. For the induction step, assume we know $f$ for all points of weight up to $r$
for some $r \geq 2s$. Consider a point $x$ with $\wt(x) = r+1$. The
set $N_r(x)$ of weight-$r$ neighbors of $x$ has size $r+1 \geq 2s +1$. Hence 
\begin{align}
\label{eq:maj-rule} 
f(x) = \Maj{y \in N_r(x)}{f(y)}.
\end{align}
by Lemma \ref{lem:key}.
\end{Proof}

Note that by Equation \ref{eq:maj-rule}, we only need to know $f$ on
the sphere of radius $r$ rather than the entire ball to compute $f$ on
inputs of weight $r+1$. This observation leads to the following
sharpening for $s \leq n/4$.

\begin{Cor}
\label{cor:sphere}
Let $s \leq n/4$. Every sensitivity $s$ function is uniquely
specified by its values on a sphere of radius $2s$.
\end{Cor}
\begin{Proof}
As before we may assume that $x =0^n$.  By Equation
\ref{eq:maj-rule}, the values of $f$ on $\cS(0^n,r)$
fix the values at $\cS(0^n,r+1)$. Hence knowing $f$ on
$\cS(0^n,2s)$ suffices to compute $f$ at points of weight $2s+1$
and beyond.  In particular, the value of $f$ is fixed at all points of weight
$n/2$ through $n$ (since $2s \leq n/2$).  Hence the value of $f$ is fixed at all points of
the ball $\cB(1^n,2s)$, and now Theorem~\ref{thm:ball} finishes the proof.
\ignore{
We will similarly show that knowing $f$ on $\cS(0^n,2s)$ specifies
it uniquely at points of Hamming weight $r \leq 2s$. The proof is
by induction on $2s -r$. The base case $r =2s$ is
trivial. For $r \leq 2s -1$, note that for $x \in \zo^n$ such that $\wt(x)
=r$, 
\[ |N_{r+1}(x)| \geq n -r \geq 2s +1 \]
since $n \geq 4s$. Hence by Lemma \ref{lem:key},
\begin{align}
\label{eq:maj-rule2} 
f(x) = \Maj{y \in N_{r+1}(x)}{f(y)}.
\end{align}}
\end{Proof}

\subsubsection{Upper and lower bounds on ${\cal F}(s,n).$}

\ignore{
\rnote{Note to future self:  If we can't or don't prove that any two distinct trees in the support of the distribution compute distinct functions, replace all the bottleneck-free DT silliness with the simpler argument
where the top $s-1$ layers of the DT are just computing a junta.  In fact, probably do that regardless since it gives a nearly-as-good bound and is simpler.}
\blue{In terms of exposition of the paper I don't know how much of this section we want to keep in the pre-appendix part.  One option would be to move it all to the appendix and replace it with something like the following:

We have the following upper and lower bounds on $|{\cal F}(s,n)|$:
\begin{Thm}
\begin{itemize}
\item ${n \choose s} 2^{2^s-1} \leq |{\cal F}(s,n)| \leq 2^{{n \choose \leq 2s}}$ for $n/4 \leq s \leq n$;
\item $2^{\Omega(2^s)} \cdot n^{2^{s-1}} \leq |{\cal F}_{s,n}| \leq 2^{{n \choose \leq 2s}}$ for $s \leq n/4$;
\item $
 n^{2^s-1}\cdot 2^{2^{s-1}-1} \leq {\cal F}(s,n) \leq n^{s \cdot 2^{2s-1}} \cdot 2^{2^{O(s^2)}}$ for $s$ such that $s \cdot 2^{2s-1} \leq n$.
 \end{itemize}
\end{Thm}
These bounds are proved in Appendix Z.  It would be interesting to determine the right form of the bound in the regime $\Omega(\log n) \leq s \leq n/4.$ 
}

\medskip
}

Recall that $|\calF(s,n)|$ denotes the number of distinct Boolean functions on $n$
variables with sensitivity at most $s$. We use the notation
${n \choose \leq k}$ to denote $\sum_{i=0}^k{n \choose i}$, the cardinality of a Hamming ball of radius $k$.

As an immediate corollary of Theorem \ref{thm:ball}, we have the following upper bound:

\begin{Cor}
\label{cor:counting}
For all $s \leq n$, we have
$|\calF(s,n)| \leq 2^{n \choose \leq 2s}.$
\end{Cor}

We have the following lower bounds:

\begin{Lem}
\label{lem:count-lower}
For all $s \leq n$, we have
$|\calF(s,n)| \geq \max \left( {n \choose s}2^{2^s -1}, (n - s
+1)^{2^{s-1}}\right).$
\end{Lem}
\begin{Proof}
The first bound comes from considering $s$-juntas. We claim that there are at least $2^{2^s -1}$ functions on $s$
variables that depend on all $s$ variables. For any function
$f:\zo^s \rgta \zo$ on $s$ variables, either $f$ or $f' =
f\oplus \prod_{i=1}^sx_i$ is sensitive to all $s$ variables. This is
because $f \oplus f' = \prod_{i=1}^sx_i$, hence one of them has full
degree as a polynomial over $\F_2$, and hence must depend on all $n$
variables. The bound now follows by considering all subsets of
$n$ variables. 

The second bound comes from the addressing functions.
Divide the variables into $s -1$ address variables $y_1,\ldots,y_{s-1}$
and $n -s +1$ output variables $x_1,\ldots,x_{n -s + 1}$. Consider the
addressing function computed by a decision tree with nodes at the first $s-1$ levels labelled
by $y_1,\ldots,y_{s-1}$ and each leaf  labelled by some $x_i$ (the
same $x_i$ can be repeated at multiple leaves). It is easy to check that this defines a family of sensitivity $s$
functions, that all the functions in the family are distinct, and that the
cardinality is as claimed.
\end{Proof}

In the setting when $s = o(n)$, the gap between our upper and lower
bounds is roughly $2^{n^s}$ versus $n^{2^s}$. The setting where $s =O(\log(n))$
is particularly intriguing.

\begin{Prob}
\label{prob:1}
Is $|\calF(2\log(n), n)| = 2^{n^{\omega(1)}}$?
\end{Prob}

\subsection{Noise Stability}

We start by showing that functions with small sensitivity satisfy a
strong noise-stability guarantee. 

For a point $x \in \zo^n$ and $\delta \in [0,1]$, let $\mathrm{N}_{1-2\delta}(x)$ denote the $\delta$-noisy version of $x$,
i.e. a draw of $y \sim \mathrm{N}_{1-2\delta}(x)$ is obtained by independently
setting each bit $y_i$ to be $x_i$ with probability $1-2\delta$ and uniformly random
with probability $2\delta$. 
The noise sensitivity of $f$ at $x$ at noise rate $\delta$, denoted $\nsf(x)$, is defined as 
\[ \nsf(x) = \Pr_{y \sim \mathrm{N}_{1-2\delta}(x)}[f(x) \neq f(y)].\] 
The noise sensitivity of $f$ at noise rate $\delta$,
denoted $\nsf$, is then defined as
\[ \nsf = \Ex_{x \sim \zo^n}[\nsf(x)] = \Pr_{x \sim \zo^n, y \sim \mathrm{N}_{1-2\delta}(x)}[f(x) \neq f(y)].\] 

The next lemma shows that low-sensitivity functions are noise-stable at every point
$x \in \zo^n$:
\begin{Lem}
\label{lem:noise-stab}
Let $f:\zo^n \to \zo$ have sensitivity $s$. For every $x \in \zo^n$ and $0 \leq \delta \leq 1/2$, we have
$\nsf(x) < 2\delta s$.
\end{Lem}
\begin{Proof}
Let $t \in [n]$. Consider a random process that starts at $x$ and then flips a uniformly
random subset $T \subseteq [n]$ of coordinates of cardinality $t$,
which takes it from $x$ to $y \in \zo^n$. We claim that  
$\Pr_{T}[f(x) \neq f(y)] \leq \frac{st}{n-t+1}.$
To see this, we can view going from $x$ to $y$ as a walk where at
each step, we pick the next coordinate to walk along uniformly from the set
of coordinates that have not been selected so far. Let $x =x_0, x_1,
\ldots, x_t = y$ denote the sequence of vertices visited by this walk. 
At $x_i$, we choose the next coordinate to flip uniformly from a set
of size $n - i$. Since $x_i$ has at most $s$ sensitive
coordinates, we have
$\Pr[f(x_i) \neq f(x_{i+1})] \leq \frac{s}{n - i}.$
Hence by a union bound we get
\begin{align*} 
\Pr[f(x_0) \neq f(x_t)] \leq \sum_{i=0}^{t-1}  \Pr[f(x_i) \neq
  f(x_{i+1})] \leq \sum_{i=0}^{t-1}\frac{s}{n - i} \leq \frac{st}{n-t+1}
\end{align*}
as claimed.

Now we turn to noise sensitivity.
We can view a draw of $y \sim \mathrm{N}_{1-2\delta}(x)$ as first choosing the number $t$ of
coordinates of $x$ to flip according to the binomial distribution
$\Bin(n,\delta)$, and then flipping a random set $T \subseteq [n]$ of size $t$.
From above, we have
$\Pr[f(y) \neq f(x)\ | \ |T| = t] \leq{\frac{st}{n-t+1}}.$
Hence
\begin{align*}
\Pr[f(x) \neq f(y)] & \leq \Ex_{t\sim \Bin(n,\delta)}\left[{\frac{st}{n-t+1}}\right]
 = s \sum_{t=1}^n \delta^t (1-\delta)^{n-t} {n \choose t} \cdot {\frac t {n-t+1}}\\
&= s \sum_{t=1}^n \delta^t (1-\delta)^{n-t} {n \choose t-1}\\
& = {\frac {s \delta}{1-\delta}} \sum_{t'=0}^{n-1} \delta^{t'}(1-\delta)^{n-t'}{n \choose t'}\\
& = {\frac {s \delta}{1-\delta}} (1-\delta^n) 
\end{align*}
which is less than $2 \delta s$ for $\delta \leq 1/2.$
\end{Proof}

\ignore{
\begin{Cor}
For any sensitivity $s$ function $f$, $\nsf \leq 3s\delta$. 
\end{Cor}
}
We can restrict the noise distribution 
and get similar bounds. 
The setting that we now describe, where we only allow walks in the lower shadow of a vertex, will
be useful later when we construct shallow formulas for a low sensitivity function $f$. 

Let  $D(x,t)$ denote the points in the lower shadow of $x$ at
distance $t$ from it (so a point in $D(x,t)$  is obtained by flipping
exactly $t$ of the bits of $x$ from 1 to 0). We show that a random
point in $D(x,t)$ is likely to agree with $x$ (for $t \leq \wt(x)/2s$).  

\begin{Lem}
\label{lem:down}
Let $\wt(x) = d \geq s$. Then if $s(f) \leq s$, we have
$\Pr_{y \in D(x,t)}[f(x)   \neq f(y)] \leq \frac{st}{d-t}.$
\end{Lem}
\begin{Proof}
We consider a family of random walks that we call \emph{downward walks}.  In such a walk,
at each step we pick a random index that is currently $1$ and set it to $0$. 
Consider a downward walk of length $t$ and let $x =x_0, x_1, \ldots,
x_t = y$ denote the sequence of vertices that are visited by the walk. We claim that 
$
\Pr[f(x_i) \neq f(x_{i+1})] \leq \frac{s}{d-i}. 
$
To see this observe that out of the $d- i$ possible $1$ indices in
$x_i$ that could be flipped to $0$, at most $s$ are sensitive. Hence
we have
\[ \Pr[f(x_0) \neq f(x_t)] \leq \sum_{i=0}^{t-1} \Pr[f(x_i) \neq f(x_{i+1})] 
\leq \sum_{i=0}^{t-1} \frac{s}{d-i} \leq \frac{st}{d-t} \]
Since $y=x_t$ is a random point
 in $D(x,t)$,  the proof is complete. 
\end{Proof}

\begin{Cor}
\label{cor:down}
Let $\wt(x) =d$ and $t \leq d/(10s+1)$. Then 
$ \Pr_{y \in D(x,t)}[f(y) \neq f(x)] \leq 1/10.$
\end{Cor}



\subsection{Bias and Interpolation} \label{sec:bias}

It is known that low sensitivity functions cannot be highly biased. For $f:\zon \to \zo$, let 
\begin{gather*}
\mu_0(f)  = \Pr_{x \in \zon}[f(x) = 0], \ \mu_1(f)  = \Pr_{x \in \zon}[f(x) = 1],\\
\mu(f)  = \min(\mu_0(f), \mu_1(f))
\end{gather*}

\begin{Lem} \label{lem:bias}
For $f:\zon \to \zo$ we have
\begin{align*}
s_0(f) & \geq \log_2\left(\fr{\mu_0(f)}\right) \ \text{if} \ \mu_0(f) > 0\\
s_1(f) & \geq \log_2\left(\fr{\mu_1(f)}\right) \ \text{if} \ \mu_1(f) > 0.
\end{align*}
Equality holds iff the set $f^{-1}(b)$ is a subcube.
\end{Lem}

We note that these bounds are implied by the classical isoperimetric
inequality, which in fact shows that $\E_{x \in f^{-1}(b)}[s(f,x)] \geq
\log(1/\mu_b(f))$ for $b=0,1$. We present a simple inductive proof of
the max-sensitivity bounds given by Lemma \ref{lem:bias} in the appendix.

We say that a set $K \subseteq \zo^n$ \emph{hits} a set of functions $\calF$
if for every $f \in \calF$, there exists $x \in K$ such that $f(x)
\neq 0$. We say that $K$ \emph{interpolates} $\calF$ if for every $f_1 \neq
f_2 \in \calF$, there exists $x \in K$ such that $f_1(x) \neq f_2(x)$.

\begin{Cor}
\label{cor:interpolate}
Let $k \geq C2^{2s}{n\choose \leq 4s}$, and let $S$ be a random subset of $\{0,1\}^n$ obtained by taking $k$
points drawn uniformly from $\{0,1\}^n$ with replacement.
The set $S$ interpolates $\calF(s,n)$ with probability $1 -\exp(-{n \choose \leq 4s})$ (over the choice of $S$) .
\end{Cor}
\begin{Proof}
We first show that large sets hit  $\calF(t,n)$ with very
high probability.  Fix $f \in \calF(t,n)$. Since we have $\mu_1(f) \geq 2^{-t}$ by Lemma \ref{lem:bias}, the
probability that $k$ random points all miss $f^{-1}(1)$ is bounded by
$(1 -2^{-t})^k \leq \exp(-k/2^t).$
By Corollary \ref{cor:counting} we have $ \calF(t,n) \leq 2^{n \choose \leq 2t}$, so
by the union bound, the probability that $S$ does not hit every function in
this set is at most 
$2^{n \choose \leq 2t}\exp(-k/2^t)$,
which is $\exp(-{n \choose \leq 2t})$ for 
$k \geq C2^t{n \choose \leq 2t}.$

Next, we claim that if $S$ hits $\calF(2s,n)$ then it interpolates $\calF(s,n)$. Given functions $f_1, f_2 \in
\calF(s,n)$, let $g = f_1 \oplus f_2$. It is easy to see that $g \in
\calF(2s,n)$. and that $g^{-1}(1)$ is the set of points $x$ where $f_1(x) \neq
f_2(x)$, so indeed if $S$ hits $\calF(2s,n)$ then it interpolates $\calF(s,n)$.  Given this,  the corollary follows from our lower bound on $k$,\ignore{Equation \eqref{eq:bound-k}, } taking $t=2s.$
\end{Proof}

\ignore{
\pnote{Keep this?}
\blue{ Since Corollary \ref{cor:interpolate} is analogous to interpolation for
low-degree polynomials, it is natural to ask for an efficient
algorithm that recovers any low-sensitivity function from a random
subset (in other words, an algorithm for exactly learning a
low-sensitivity function from  uniform random labeled examples).  }
}

\section{Efficient algorithms for computing low sensitivity functions}

\subsection{Small circuits}

In this subsection, we will prove Theorem \ref{thm:ckt}.  Recall that
the proof of Theorem \ref{thm:ball} gives an algorithm to compute the
truth table of $f$ from an advice string which tells us the values on
some Hamming ball of radius $2s +1$. In this section we present two
algorithms which, given this advice, can (relatively) efficiently compute any entry of the truth table without
computing the truth-table in its entirety. This is equivalent to a
small circuit computing $f$. We first give a (non-uniform) ``bottom-up'' algorithm for
computing $f$ at a given input point $x \in \{0,1\}^n$. 
In the appendix we describe a ``top-down'' algorithm  with a similar performance bound.

\subsubsection{A Bottom-Up Algorithm}

The algorithm takes as advice the values of $f$ on $\cB(0^n,2s)$. It
then shifts the center of the ball along a shortest path from $0^n$ to $x$,
computing the values of $f$ on the shifted ball at each step. This
computation is made possible by a lemma showing that when we shift a
Hamming ball $\cB$ by a unit vector to get a new ball $\cB'$, points in
$\cB'$ either lie in $\cB$ or are adjacent to many points in $\cB$,
which lets us apply Lemma \ref{lem:key}.  

Let $\ind(S)$ denote the indicator of $S \subseteq [n]$ and $S \Delta T$ denote the
symmetric difference of the sets $S, T$.  For $B \subseteq \{0,1\}^n$ we write $B \oplus e_i$
to denote the pointwise shift of $B$ by the unit vector $e_i$.

\begin{Lem}
\ignore{Let $\cB(x \oplus e_i,r)$ be the shift of the Hamming
ball $\cB(x,r)$ by the unit vector $e_i$.}For any $y \in \cB(x \oplus
e_i,r)\setminus \cB(x,r)$, we have $|N(y) \cap \cB(x,r)| = r+1$.
\end{Lem}
\begin{Proof} Fix any such $y$.  Since $\cB(x \oplus e_i,r) = \cB(x,r) \oplus e_i$, we have that
\begin{align*} 
y & = x' \oplus e_i \ \text{for some} \ x' \in \cB(x,r), \text{~where}\\ 
x'& = x \oplus \ind(S) \ \text{for some} \ S\subseteq [n], |S| \leq r, \text{~and hence}\\
y & = x \oplus \ind(S \Delta \{i\}). 
\end{align*}

If $i \in S$ or $|S| \leq r-1$, then $|S \Delta \{i\}| \leq r$; but this means that 
$y \in \cB(x,r)$, which is in contradiction to our assumption that $y \in \cB(x \oplus
e_i,r)\setminus \cB(x,r)$.  Hence  $i \not\in S$ and $|S| =r$. But then we have $y \oplus e_j \in \calB(x,r)$
for precisely those $j$ that belong to $S \cup \{i\}$, which gives the claim.
\end{Proof}

\begin{Cor}
\label{cor:shift}
Knowing the values of $f$ on $\cB(x,2s)$ lets us compute $f$ on $\cB(x \oplus e_i,2s)$.
\end{Cor}
\begin{Proof}
Either $y \in \cB(x \oplus e_i,2s)$ lies in $\cB(x,2s)$ so we know
$f(y)$ already, or by the previous lemma $y$ has $2s +1$ neighbors in $\cB(x,2s)$, in which
case Lemma \ref{lem:key} gives
$f(y) = \Maj{y' \in N(y) \cap \cB(x,2s)}{f(y')}.$
\end{Proof}

Now we can give our algorithm for computing $f(x)$ at an
arbitrary input $x \in \zo^n$. 

\medskip

\myalgo{Bottom-Up}{
{\bf Advice: }The value of $f$ at all points in $\cB(0^n,2s)$.\\
{\bf Input: }$x \in \zo^n$.

\begin{enumerate}
\item Let $0^n =x_0, x_1,\ldots, x_d =x$ be a shortest
  path from $0^n$ to $x$.
\item For $i \in \{1,\ldots,d\}$ compute $f$ on
  $\cB(x_i,2s)$ using the values at points in $\cB(x_{i-1},2s)$.
\item Output $f(x_d)$.
\end{enumerate}
}

\medskip

\begin{Thm}
\label{thm:algP}
The algorithm \emph{{\bf Bottom-Up}} computes $f(x)$ for any input $x$ in time $O(sn^{2s+1})$ using space
$O(n^{2s})$.
\end{Thm}
\begin{Proof}
The values at $\cB(0^n,2s)$ are known as advice. Corollary
\ref{cor:shift} tells us how to compute the values at $\cB(x_i,2s)$ using the values on $\cB(x_{i-1},2s)$.
If we store the values at $\cB(x_{i-1},2s)$ in an array indexed by
subsets of size $2s$, the value at any point $y \in \cB(x_i,2s)$ can be
computed in time $O(s)$, by performing $2s +1$ array lookups and then
taking the majority. Thus computing the values over the entire ball takes time
$O(sn^{2s})$, and we repeat this $d \leq n$ times. Finally, at stage $i$ we only need to
store the values of $f$ on the latest shift, $\cB(x_{i-1},2s)$,  so the total space required is $O(n^{2s})$.
\end{Proof}


\subsection{Small-depth Formulas} \label{sec:formulas}

Theorem \ref{thm:algP} established that any $n$-variable sensitivity-$s$ function $f$ is computed by a 
circuit of size $O(sn^{2s+1})$, but of relatively large depth
$O(n^{2s})$.  In this section we improve this depth by showing
that \emph{shallow} circuits of essentially the same size
(equivalently, formulas of small depth) can compute low-sensitivity functions.

For $\mu  <1/2$, let $B(c,\mu)$ denote the product distribution over $y
\in \zo^c$  where $\Pr[y_i =1 ] = \mu$ for each $i \in [c].$
For constants $1/2 > \mu > \delta >0$, let $c=c(\mu,\delta) \in \Z$ be the smallest integer constant such that 
\[\Pr_{y \sim B(c,\mu)}[\Maj{i \in [c]}{y_i} = 1] \leq \delta.\]

We now present a randomized parallel algorithm for computing $f(x).$

\medskip

\myalgo{Parallel-Algorithm}{
{\bf Advice: }$f$ at all points in $\cB(0^n,10s)$.\\
{\bf Input: }$x \in \zo^n$.

Let $d = \wt(x)$, $t = \lfloor d/(10s+1)\rfloor$, $c = c(1/5,1/20)$.
\begin{enumerate}
\item If $d \leq 10s$, return $A(x) = f(x)$.
\item Else sample $y_1,\ldots,y_c$ randomly from $D(x,t)$. Recursively run
{\bf Parallel-Algorithm} to compute
  $A(y_i)$ in parallel for all $i \in [c]$.
\item Return $A(x) = \Maj{i \in [c]}{A(y_i)}$.
\end{enumerate}
}

\medskip

For brevity we use $A$ to denote the algorithm above and $A(x) \in \zo$ to 
denote the random variable which is its output on input $x$. For $d \geq
10s +1$, the random choices of $A$ in computing $A(x)$ are described
by a $c$-regular tree.  The tree's root is labeled by $x$ and its children are labeled by $y_1,\dots,y_c$;
its leaves are labeled by strings that each have Hamming weight at most
$10s$. Further, the various subtrees rooted at each level are independent of
each other.

\begin{Thm}
\label{thm:parallel}
The algorithm runs in parallel time $O(s\log n)$ using $n^{O(s)}$ processors.
For any $x \in \zo^n$, we have
$\Pr_A[A(x) \neq f(x)] \leq \frac{1}{20},$
where $\Pr_A$ denotes that the probability is over the random coin
tosses of the algorithm.
\end{Thm}
\begin{Proof}
We first prove the correctness of the algorithm by induction on $\wt(x) =d$. When $d \leq 10s$, the
claim follows trivially. Assume that the claim is true for $\wt(x)
\leq d -1$, and consider an input $x$ of weight $d$. Note that every
$y \in D(x,t)$ has $\wt(y) = d -t \leq d-1$, hence the inductive
hypothesis applies to it. For each $i \in [c],$ we independently have
\[ \Pr_{A}[A(y) \neq f(x)] \leq \Pr_{A,y_i}[A(y_i) \neq f(y_i)] + \Pr_{y_i \in D(x,t)}[f(y_i) \neq f(x)]
\leq \frac{1}{10} + \frac{1}{20} < \frac{1}{5}.\]
where the $1/10$ bound is by Corollary \ref{cor:down} and the
$1/20$ is by the inductive hypothesis. The algorithm samples $c$ independent points  $y_i \in D(x,t)$,
computes $A(y_i)$ for each of them using independent randomness, and
then returns the majority of $A(y_i)$ over those $i \in [c]$. Hence,
by our choice of $c=c(1/5,1/20),$ we have that
$\Pr_{A}[\Maj{i \in [c]}{A(y_i)} \neq f(x)] \leq \frac{1}{20}.$

To bound the running time, we observe that for $d \geq 10s +1$, 
\begin{align*} 
t = \left\lfloor \frac{d}{10s+1}\right\rfloor \geq \frac{d}{25s}, \quad \text{so} \quad 
d -t \leq d\left(1  - \frac{1}{25s}\right).
\end{align*}
But this implies that in $k = O(s\log d)$ steps, the weight reduces below
$10s +1$. The number of processors required is bounded by $c^k =  n^{O(s)}$. 
\end{Proof}

By hardwiring the random bits and the advice bits, we can conclude
that functions with low sensitivity have small-depth formulas, thus
proving Theorem \ref{thm:formulas}.

\section{Self-correction}

In this section we show that functions with low sensitivity admit self-correctors. 
Recall that for Boolean functions, $f, g:
\zo^n \to \zo$ we write $\delta(f,g)$ to denote $\Pr_{x \in \zo^n}[f(x) \neq g(x)].$

Our self-corrector is given a function $r:\zo^n \to \zo$ such that
there exists $f \in \calF(s,n)$ satisfying
$\delta(r,f) \leq 2^{-cs}$ for some constant $c > 2$ to be specified
later. By Lemma \ref{lem:bias}, it follows that any two
sensitivity $s$ functions differ in at least $2^{-2s}$ fraction of points, so
if such a function $f$ exists, it must be unique. We consider two
settings (in analogy with coding theory): in the global setting, the
self-corrector is given the truth-table of $r$ as input and is required to
produce the truth-table of $f$ as output. In the local setting, the
algorithm has black-box oracle access to $r$. It is given $x \in \zo^n$ as
input, and the desired output is $f(x)$. 

At a high level, our self-corrector relies on the fact that
small-sensitivity sets are noise-stable at noise rate $\delta  \approx
1/s$, by Lemma \ref{lem:noise-stab}, whereas small sets of density
$\mu \leq  c^{-s}$ tend to be noise sensitive. The analysis uses the
hypercontractivity of the $T_{1-2\delta}(\cdot)$ operator.

Following \cite{Ryan}, for $f: \{0,1\}^n \to \R$, we define
\[ \trf(x) = \Ex_{y \sim \mathrm{N}_{1-2\delta}(x)}[f(y)],\]
where recall that a draw of $y \sim \mathrm{N}_{1-2\delta}(x)$ is obtained by independently
setting each bit $y_i$ to be $x_i$ with probability $1-2\delta$ and uniformly random
with probability $2\delta$.
We can view $(x,y)$ where $x \sim \zo^n$ and $y \sim \mathrm{N}_{1-2\delta}(x)$ as
defining a distribution on the edges of the complete graph on the
vertex set $\zo^n$. We refer to this weighted graph as the $\delta$-{\em noisy
  hypercube}. The $(2,q)$-Hypercontractivity Theorem quantifies the
expansion of the noisy hypercube:

\begin{Thm}
\label{thm:hc}
($(2,q)$-Hypercontractivity.) 
Let $f: \zo^n \to \R$. Then 
\[ \| \trf\|_q \leq \|f\|_2 \ \ \text{for} \ \ 2 \leq  q \leq 1 + \frac{1}{(1-2\delta)^2}.\]
\end{Thm}

We need the following consequence, which says that for any small set $S$, most points
do not have too many neighbors in the noisy hypercube that lie within
$S$. For $S \subseteq \zo^n$, let us define the
set $\Lambda_{\delta,\theta}(S)$ of those points for which a $\theta$ fraction of neighbors in the
$\delta$-noisy hypercube lie in $S$. Formally,
\[ \Lambda_{\delta,\theta}(S) = \{x \in \zo^n \ s.t. \ \Pr_{y \sim \mathrm{N}_{1-2\delta}(x)}[y \in S] \geq \theta\}. \]
 
 Abusing the notation from Section \ref{sec:bias}, for $S \subseteq \{0,1\}^n$ we write $\mu(S)$ to denote
 $\Pr_{x \in \{0,1\}^n}[x \in S].$

\begin{Lem}
\label{lem:sse}
We have
\[ \mu(\Lambda_{\delta,\theta}(S)) \leq \bkets{\frac{\mu(S)}{\theta^2}}^{1 +2 \delta}. \]
\end{Lem}

\begin{Proof}
Let $f(x) = \ind(x \in S)$. Then 
\[ \trf(x) = \Pr_{y \in \mathrm{N}_{1-2\delta}(x)}[y \in S]. \]
Hence $\Lambda_{\delta,\theta}(S)$ is the set of those $x$ such that $\trf(x)
\geq \theta$. 

Let $q = 2(1 + 2\delta)$. It is easy to see that $q$ satsfies the
hypothesis of Theorem \ref{thm:hc}. Hence we can bound the $q^{th}$ moment of $\trf$ as
\[\Ex_{x \in \zo^n}[(\trf(x))^q] \leq \|f\|_2^q = \mu(S)^{q/2}.\]
Hence by Markov's inequality,
\[ \Pr_{x \in \zo^n}[\trf(x) \geq \theta] \leq \frac{\mu(S)^{q/2}}{\theta^q}.\]
The claim follows from our choice of $q$.
\end{Proof}

\begin{Cor}
\label{cor:sse}
If $\mu(S) \leq \theta^{4 + 2/\delta}$, then $\mu(\GS) \leq \mu(S)^{1 + \delta}$. 
\end{Cor}
\begin{Proof}
By Lemma \ref{lem:sse}, it suffices that 
$ \bkets{\frac{\mu(S)}{\theta^2}}^{1 +2 \delta} \leq \mu(S)^{1
  +\delta}$,
and it is easy to check that this condition holds for our choice of
$\mu(S)$. 
\end{Proof}

\subsection{Global Self-correction}

\ignore{
For Boolean functions, $f, g:
\zo^n \to \zo$ let 
\[ \delta(f,g) = \Pr_{x \in \zo^n}[f(x) \neq g(x)].\]
}
Our global self-corrector is given a function $r:\zo^n \to \zo$ such that
there exists $f \in \calF(s,n)$ satisfying
$\delta(r,f) \leq 2^{-c_1s}$ for some constant $c_1 > 2$ to be specified
later. By Lemma \ref{lem:bias}, it follows that any two
sensitivity $s$ functions differ in at least $2^{-2s}$ fraction of points, so
such a function $f$ if it exists must be unique. Our self-corrector
defines a sequence of functions $f_0,\ldots,f_T$ such that $f_0 =r$ and
$f_T = f$ (with high probability).

\myalgo{Global Self-corrector}{
{\bf Input: }$r:\zo^n \to \zo^n$ such that $\delta(r,f) \leq 2^{-c_1 s}$ for some 
$f \in {\cal F}(s,n)$.\\
{\bf Output:} The sensitivity-$s$ function $f$.\\

Let $f_0 =r$, $k = c_2s\log(n/s)$, $\delta = 1/(20s)$.\\
For $t = 1,\ldots, k$,\\
$~~~$ For every $x \in \zo^n$, \\
$~~~~~~$ Let $f_t(x) = \Maj{y \sim \mathrm{N}_{1-2\delta}(x)}{f_{t-1}(y)}$.\\
Return $f_k$.
}

The algorithm runs in time $2^{O(n)}$, which is
polynomial in the length of its output (which is a truth table of size
$2^n$). To analyze the algorithm, let us define the sets $S_t$ for $t
\in \{0,\ldots,T\}$ as 
\[ S_t = \{x \in \zo^n \ \text{such that} \ f_t(x) \neq f(x).\}\]

The following is the key lemma for the analysis.

\begin{Lem}
\label{lem:self-correct}
We have $S_t \subseteq \Lambda_{\delta,2/5}(S_{t-1})$. 
\end{Lem}
\begin{Proof}
For $x \in S_t$, 
\[ f(x) \neq \Maj{y \sim \mathrm{N}_{1-2\delta}(x) }{f_{t-1}(y)},\]
hence
\[ \Pr_{y \sim \mathrm{N}_{1-2\delta}(x)}[f(x) \neq f_{t-1}(y)] \geq \frac{1}{2}. \]
We can upper bound this probability by
\begin{align*} 
\Pr_{y \sim \mathrm{N}_{1-2\delta}(x)}[f(x) \neq f_{t-1}(y)] & \leq \Pr_{y \sim \mathrm{N}_{1-2\delta}(x)}[f(x)
  \neq f(y)] + \Pr_{y \sim \mathrm{N}_{1-2\delta}(x)}[f(y) \neq f_{t-1}(y)].
\end{align*}
Since the distributions $\noise_{\delta}(x)$ and $\mathrm{N}_{1-2\delta}(x)$ are identical, we can bound the first term by Lemma \ref{lem:noise-stab}, which gives
\[
\Pr_{y \sim \mathrm{N}_{1-2\delta}(x)}[f(x) \neq f(y)] \leq 2 s \delta = \frac{1}{10}. 
\]
Hence
\[ \Pr_{y \sim \mathrm{N}_{1-2\delta}(x)}[f(y) \neq f_{t-1}(y)] \geq \frac{1}{2} - \frac{1}{10} = \frac{2}{5}.\]
But $f(y) \neq f_{t-1}(y)$ implies $y \in S_{t-1}$, hence  by
definition of $\GS$ we have $x \in \Lambda_{\delta,2/5}(S_{t-1})$. 
\end{Proof}

We can now analyze our global self-corrector.

\begin{Thm}
\label{thm:global}
There exist constants $c_1, c_2$ such that if $\delta(r,f) \leq 2^{-c_1
  s}$ for some $f \in \calF(s,n)$, then for $k \geq c_2s\log(n/s)$, we have $f_k =f$.
\end{Thm}  
\begin{Proof}
Let $\delta =1/(20s)$. Assume that there exists $f \in \calF(s,n)$
such that 
\[ \delta(f,s) = \mu(S_0) \leq 2^{-c_1s} \leq (2/5)^{4 + 40s}.\] 
By Lemma \ref{lem:self-correct} and Corollary \ref{cor:sse}, we have 
\[ \mu(S_t) \leq \mu(\Lambda_{\delta,2/5}(S_{t-1})) \leq
\mu(S_{t-1})^{1 + \delta} \leq \mu(S_0)^{(1 + \delta)^t}. \]
For $t \geq c_2'\ln(n/s)/\delta = c_2s\log(n/s)$, we have
\[ \mu(S_t) \leq \mu(S_0)^{(1+ \delta)^t} < 2^{-n}.\]
But since $S_t \subseteq \zo^n$, it must be the empty set, and this
implies that $f_t =f$. 
\end{Proof}

\subsection{Local Self-Correction}

Recall that in the local self-correction problem, the algorithm is given $x
\in \zo^n$ as input and oracle access to $r:\zo^n \to \zo$ such that
$\delta(r,f) \leq 2^{-d_1s}$ for some constant $d_1 > 2$ to be
specified later. The goal is to compute $f(x)$. Our local
algorithm can be viewed as derived from the global algorithm, where we
replace the Majority computation with sampling, and only compute the
parts of the truth tables that are essential to computing $f_T(x)$. 

We define a distribution $\cT_k(x)$ over $c$-regular trees of depth $k$
rooted at $x$, where each tree node is labelled with a point in
$\zo^n$. To sample a tree $T_1(x)$ from $\cT_1(x)$, we place $x$ at the root,
then sample $c$ independent points from $\mathrm{N}_{1-2\delta}(x)$, and place them at the
leaves. To sample a tree $T_k(x)$ from $\cT_k(x)$, we first sample 
$T_{k-1}(x) \sim \cT_{k-1}(x)$ and then for every leaf $x_i \in  T_{k-1}(x)$, we sample $c$
independent points according to $\mathrm{N}_{1-2\delta}(x_i)$, and make them the children of $x_i$.
(Note the close similarity between these trees and the trees discussed in Section \ref{sec:formulas}.  
The difference is that the trees of Section \ref{sec:formulas} correspond to random walks that are constructed to go downward while now the random walks corresponding to the noise process $\mathrm{N}_{1-2\delta}(\cdot)$ do 
not have this constraint.)

Given oracle access to $r: \zo^n \to \zo$, we use the tree $T_k(x)$ to compute a
guess for the value of $f(x)$, by querying $r$ at the leaves
and then using Recursive Majority. In more detail, we
define functions $\tr_0,\tr_1,\dots,\tr_k$ which collectively assign a {\em guess} for every
node in $T_k$. (In more detail, each $\tr_i$ is a function from $L(T_k(x),i)$ to $\{0,1\}$, where $L(T_k(x),i)$ is the set of points in $\{0,1\}^n$ that are at the nodes at depth $k-i$ in $T_k(x).$)  For each leaf $y$, we let $\tr_0(y) =
r(y)$. Once $\tr_{k-t}$ has been defined for nodes at depth
$t$ in $T_k(x)$, given $y$ at depth $t-1$ in $T_k(x)$, we set $\tr_{k-t+1}(y)$ to be the
majority of $\tr_{k-t}$ at its children. We output $\tr_k(x)$ as our estimate for
$f(x)$.

\myalgo{Local Self-corrector}{
{\bf Input:}  $x \in \zo^n$, oracle for $r:\zo^n \to \zo$ such that $\delta(r,f)
\leq 2^{-d_1s}$ for some $f \in \calF(s,n)$.\\
{\bf Output:} $b \in \zo$ which equals $f(x)$ with probability $1 -\eps$. \\

Let $\delta = 1/(20s)$, $c= c(1/4,\eps)$, $k \in \Z$. \\
Sample $T_k \sim \cT(k,x)$. \\
For each leaf $y \in T_k$, query $r(y)$. \\
For $i=0$ to $k$, compute $\tr_i:L(T_k(x),i) \to \zo$ as described above.\\
Output $\tr_k(x)$.
}

To analyze the algorithm, for $k \in \Z$ define
\[ S_k= \cbkets{x \in \zo^n \ \text{such that} \ \Pr_{T_k(x) \sim \cT_k(x)}[\tr_k(x) \neq f(x)] > \eps}.\]

The following is analogous to Lemma \ref{lem:self-correct}:
\begin{Lem}
\label{lem:self-correct2}
For $k \geq 1$ and $\eps<1/25$, we have $S_k \subseteq \Lambda_{\delta,1/10}(S_{k-1})$. 
\end{Lem}
\begin{Proof}
We have $\tr_k(x)= \Maj{1 \leq i  \leq c}{b_i}$ where each $b_i$ is
drawn independently according to $\tr_{k-1}(\mathrm{N}_{1-2\delta}(x))$. If $x \in
S_k$, then by our choice of $c = c(1/4,\eps)$,
\[ \Pr_{\stackrel{y \sim \mathrm{N}_{1-2\delta}(x)}{T_{k-1}(y) \sim \cT_{k-1}(y)}}[\tr_{k-1}(y)
  \neq f(x)] > \frac{1}{4}. \]
On the other hand, we also have
\[\Pr_{\stackrel{y \sim \mathrm{N}_{1-2\delta}(x)}{T_{k-1}(y) \sim \cT_{k-1}(y)}}[\tr_{k-1}(y)
  \neq f(x)] \leq 
\Pr_{\stackrel{y \sim \mathrm{N}_{1-2\delta}(x)}{T_{k-1}(y) \sim \cT_{k-1}(y)}}[f(y) \neq f(x)] + 
\Pr_{\stackrel{y \sim \mathrm{N}_{1-2\delta}(x)}{T_{k-1}(y) \sim \cT_{k-1}(y)}}[\tr_{k-1}(y) \neq f(y)].\]
The first term on the LHS is bounded by $1/10$ by Lemma \ref{lem:noise-stab}. Hence we have
\[ \Pr_{\stackrel{y \sim N_\delta(x)}{T_{k-1}(y) \sim \cT_{k-1}(y)}}[\tr_{k-1}(y) \neq f(y)] \geq \frac{1}{4} - \frac{1}{10} > \frac{1}{7}.\]
But by the definition of $S_{k-1}$, 
\begin{align*} \Pr_{\stackrel{y \sim \mathrm{N}_{1-2\delta}(x)}{T_{k-1}(y) \sim \cT_{k-1}(y)}}[\tr_{k-1}(y) \neq
f(y)] & \leq \eps\cdot\Pr_{y \sim \mathrm{N}_{1-2\delta}(x)}\sbkets{y \not\in S_{k-1}} + \Pr_{y \sim \mathrm{N}_{1-2\delta}(x)}[y \in S_{k-1}] \\
& \leq \eps+ \Pr_{y \sim \mathrm{N}_{1-2\delta}(x)}[y \in
S_{k-1}].
\end{align*}
Hence for $\eps < 1/25$,  
\[ \Pr_{y \sim \mathrm{N}_{1-2\delta}(x)}[y \in S_{k-1}] \geq \frac{1}{7} - \eps \geq \frac{1}{10}, \]
so by the definition of $\GS$ we have $x \in \Lambda_{\delta,1/10}(S_{k-1})$. 
\end{Proof}

We can now analyze our local self-corrector.

\begin{Thm}
\label{thm:local}
There exist constants $d_1, d_2$ such that if $\delta(r,f) \leq 2^{-d_1
  s}$ for some $f \in \calF(s,n)$, then for $k \geq d_2s\log(n/s)$
 we have that $\tr_k(x) =f(x)$ with probability $0.99$. The algorithm
  queries the oracle for $r$ at $(n/s)^{O(s)}$ points.
\end{Thm}  
\begin{Proof}
Let $\delta =1/(20s)$. Let $d_1 > 0$ be such that
\[ 2^{-d_1s} < (0.1)^{4 + 60s}.\] 
Assume there exists $f \in \calF(s,n)$ such that 
\[ \delta(r,f) \leq 2^{-d_1s}. \]
Observe that $\tr_0(x) = r(x)$, so consequently $\mu(S_0) = \delta(r,f)$. 
By Lemma \ref{lem:self-correct2} and Corollary \ref{cor:sse}, we have 
\[ \mu(S_k) \leq \mu(\Lambda_{\delta,1/10}(S_{k-1})) \leq
\mu(S_{k-1})^{1 + \delta} \leq \mu(S_0)^{(1 + \delta)^k}. \]
For $k \geq d_2'\ln(n/s)/\delta = d_2s\log(n/s)$, we have
$\mu(S_t) < 2^{-n},$
so $S_t$ must be the empty set. But this implies that  $\tr_k(x) = f(x)$ except with probability $\eps$.

The number of queries to the oracle is bounded by the number of leaves
in the tree, which is $c^k$. Setting $\eps =1/100$, since $c(1/4,1/100) =O(1)$, this is at
most $c^k = (n/s)^{O(s)}$. We can amplify the success probability to
$1 - \eps$ using $c(1/100,\eps) = O(\log(1/\eps))$ independent
repetitions. 
\end{Proof}

\medskip
\noindent {\bf Discussion.}  Every real polynomial of degree $d$ computing a Boolean function is also a degree $d$
polynomial over $\F_2$. Hence, it has a natural self-corrector which
queries the value at a random affine subspace of dimension $d+1$
containing $x$, and then outputs the $\ms{XOR}$ of those
values. Conjecture \ref{conj:2} implies that this self-corrector also
works for low sensitivity functions. The parameters one would
get are incomparable to Theorem \ref{thm:local-intro}; we find it
interesting that this natural self-corrector is very different from the algorithm of Theorem 
\ref{thm:local-intro}.

We further remark that every Boolean function with real polynomial degree $\deg(f) \leq d$
satsifies $s(f) \leq O(d^2)$ (recall Theorem~\ref{thm:ns}). Thus, Theorem
\ref{thm:local-intro} gives a self-corrector for functions with
$\deg(f) \leq d$ that has query complexity $n^{O(d^2)}$. It is interesting
to note (by considering the example of parity), that this
performance guarantee does not extend to all functions  with  $\F_2$ degree $\deg_2(f) \leq d$.


\section{Propagation rules}

We have seen that low-degree functions and low-sensitivity functions share the 
property that they are uniquely specified by their values on small-radius Hamming
balls. In either case, we can use these values over a small Hamming ball to
infer the values at other points in $\zo^n$ using simple ``local
propagation'' rules. The propagation rules for the two types of functions are quite different, but Conjecture \ref{conj:2} and its converse given by Theorem
\ref{thm:ns} together imply that the two rules must converge beyond a certain
radius. In this section, we discuss this as a possible approach to Conjecture
\ref{conj:2} and some questions that arise from it. 

\subsection{Low sensitivity functions: the Majority Rule} 
If $f: \zo^n \to \zo$ has sensitivity $s$, Theorem \ref{thm:ball} implies that given the values of $f$ on
a ball of radius $2s$, we can recover $f$ at
points at distance $r +1\geq 2s + 1$ from the center by taking the Majority value over
its neighbors at distance $r$ (see Equation (\ref{eq:maj-rule})). 
It is worth noting that as $r$ gets
large, the Majority is increasingly lopsided: at most $s$ out of
$r$ points are in the minority. 
We refer to the process of inferring $f$'s values everywhere from its values on a ball via the Majority rule,
increasing the distance from the center by one at a time, as ``applying the Majority rule''. 

For concreteness, let us conder the ball centered at $0^n$.
If there exists a sensitivity $s$ function $f:
\zo^n \rgta \zo$ such that the points in $\cB(0^n,2s)$ are labelled
according to $f$, then applying the Majority 
rule recovers $f$. However, not every labeling of $\cB(0^n,2s)$ will
extend to a low sensitivity function on $\zo^n$ via the Majority
Rule. It is an interesting question to characterize such labelings;
progress here will likely lead to progress on Question
\ref{prob:1}. 
An obvious necessary condition is that every point in
$\cB(0^n,2s)$ should have sensitivity at most $s$, but this is not
sufficient. This can be seen by considering the DNF version of the ``tribes'' function,
where there are $n/s$ disjoint tribes, each tribe is
of size $s$, and $n > s^2$.  (So this function $f$ is an $(n/s)$-way OR of 
$s$-way ANDs over disjoint sets of variables.)  Every $x \in \cB(0^n,2s)$ has
$s(f,x) \leq s$ --- in fact, this is true for every $x \in \cB(0^n,s(s-1))$ --- but it can be verified that
applying the Majority rule starting from the ball of
radius $2s$ does recover the Tribes function, which has sensitivity
$n/s > s$.   
Another natural question is whether there is a nice characterization of the class of
functions that can be obtained by applying the majority rule to a labeling of
$\cB(0^n,2s)$.

\subsection{Low degree functions: the Parity Rule}

It is well known that all functions $f:\zo^n \rgta \R$ with $\deg(f)
\leq d$ are uniquely specified by their values on a ball of radius $d$. This follows from the M{\"o}bius inversion formula. 
Again, let us fix the center to be $0^n$ for concreteness. 
Letting $\ind(T)$ denote the indicator vector of the set $T$, the
formula (see e.g. Section 2.1 of \cite{Jukna:12}) states that
\begin{align}
\label{eq:Mobius}
f(x) = \sum_{S \subseteq [n]}c_S\prod_{i \in S}x_i \ \ \text{where}
\ \ c_S & = \sum_{T \subseteq S}(-1)^{|S| - |T|}f(\ind(T)).
\end{align}
From this it can be inferred that if $\deg(f) \leq d$, then for $|S| \geq d+1$, we have
\begin{align}
\label{eq:Parity-rule}
f(\ind(S)) = \sum_{T \subset S}(-1)^{|S| - |T|+1}f(\ind(T)).
\end{align}
We will refer to Equation \eqref{eq:Parity-rule} as the ``Parity rule'',
since it states that $f$ is uncorrelated with the parity of the variables
in $S$ on the subcube given by $\{\ind(T):T \subseteq S\}$. 
We refer to the process of inferring $f$'s values everywhere from its values on a ball of radius $d$ via the Parity rule,
increasing the distance from the center by one at a time, as ``applying the Parity rule''.

Given a (partial) function $f: \cB(0^n,d)\to \zo$, applying the Parity rule starting from the values of $f$ on 
$\cB(0^n,d)$ lets us extend $f$ to all of $\zo^n$. Note that the resulting total function
$f$ is guaranteed to have degree at most $d$, but it is not guaranteed to be Boolean-valued everywhere on
$\{0,1\}^n$. Indeed, an easy counting argument (see e.g. Lemma~31 of \cite{MOR+:07eccc}) shows that there are at most $2^{d^2 2^{2d}} \cdot {n \choose {d2^d}}$ degree-$d$ functions over $\{0,1\}^n$, whereas the number of partial functions $f: \cB(0^n,d)\to \zo$ is $2^{{n \choose \leq d}}$. 
It is an interesting question to characterize the set of partial functions
$f:\cB(0^n,d) \to \zo$  whose extension by the Parity rule is a Boolean function. 

On the other hand, every partial function $f: \cB(0^n,d) \rgta \zo$ can be
uniquely extended to a total function $f:\zo^n \rgta \zo$ such that
$\deg_2(f) =d$. This follows from the Mobius inversion formula for
multilinear polynomials over $\F_2$:
\begin{align}
\label{eq:Mobius-f2}
f(x) = \sum_{S \subseteq [n]}c_S\prod_{i \in S}x_i \ \ \text{where} \ \ c_S & = \sum_{T \subseteq S}f(\ind(T))
\end{align}
where the sums are modulo $2$. If $\deg_2(f) \leq d$, then $c_S =0$ for all $S$ where $|S| \geq
d +1$. Hence by Equation \eqref{eq:Mobius}, for $|S| \geq d+1$, we have the simple rule
\begin{align}
\label{eq:f2-rule}
f(\ind(S)) = \sum_{T \subset S}f(\ind(T)).
\end{align}
We can view this as a propagation rule for functions with $\deg_2(f)
\leq d$, which extends a labeling of the ball $ \cB(0^n,d)$ to the entire cube $\zo^n.$
If we start from a labeling of the ball which corresponds
to a function $f:\zo^n \to \zo$ with $\deg(f) \leq d$, then Equation
\eqref{eq:f2-rule} above coincides with the Parity rule. 

\subsection{When do the rules work?}

Given a partial function $g:\cB(x_0,r) \to \{0,1\}$, we can extend it to a
total function  $g^\Mj:\zo^n \rgta \zo$ by applying the Majority rule (if there
is not a clear majority among the neighbors queried, the value is
underetmined). We can also extend it to a total function $g^\Par:\zo^n
\rgta \R$ using the Parity rule. Given a function $f:\zo^n \to \zo$, and a center $x_0$, we define
a series of partial functions $f|_{\cB(x_0,r)}$ obtained by restricting
$f$ to the ball of radius $r$ around $x_0$. We are interested in
how large $r$ needs to be for the Parity and Majority rules to return
the function $f$ for every choice of center $x_0$ . Formally, we define the
following quantities.

\begin{Def}
Let $r^\Par(f)$ be the smallest $r$ such that for every $x_0 \in \zo^n$, the Parity rule applied to $\cB(x_0,r)$ returns the
function $f$. Formally, 
\[ r^\Par(f)= \min \{r:  \forall x_0 \in \zo^n, \ (f|_{\cB(x_0, r)})^\Par = f\}.\]
Similarly, let $r^\Mj(f)$ be the smallest $r$ such that for every $x_0
\in \zo^n$,  the Majority rule applied to $\cB(x_0,r)$ returns the
function $f$. Formally, 
\[ r^\Mj(f)= \min \{r: \forall x_0 \in \zo^n (f|_{\cB(x_0,r)})^\Mj = f\}.\]  
\end{Def}

It is easy to see that $r^\Par$ captures the real degree of $f$:
\begin{Lem}
\label{lem:par}
For all $f: \{0,1\}^n \to \{0,1\}$, we have $r^\Par(f) = \deg(f)$.
\end{Lem}
\begin{Proof}
The inequality $r^\Par(f) \leq \deg(f)$ follows from the fact that the Parity rule
correctly extends degree $d$ functions starting from any Hamming ball of radius $d$. 

On the other hand for any center $x_0$, running the Parity rule on $f|_{\cB(x_0,r)}$ for some $r <
\deg(f)$ results in a function $(f|_{\cB(x_0,r)})^\Par$ of degree at most
$r$, since the Parity rule explicitly sets the coefficients of
monomials of degree higher than $r$ to $0$. But then it follows that
$(f|_{\cB(x_0,r)})^\Par \neq f$, since their difference is a non-zero multilinear polynomial. 
\end{Proof}

The proof above shows that quantifying over $x_0$ is not necessary in
the definition of $r^\Par(f)$, since for every $x_0 \in \zo^n$, we have
\[ r^\Par(f)= \min \{r:  \ (f|_{\cB(x_0,r)})^\Par = f\}.\]

We now turn to the Majority rule. 

\begin{Lem}
\label{lem:maj}
For all $f: \{0,1\}^n \to \{0,1\}$,  we have $r^\Mj(f) = \min(2s(f),n)$. 
\end{Lem}
\begin{Proof}
We have $r^\Mj(f) \leq n$, since $\cB(x_0,n)$ is the entire
Hamming cube.  The upper bound $r^\Mj(f) \leq 2s(f)$ follows from the
definiton of the Majority rule and Theorem \ref{thm:ball}. 

For the second part, we show that for every $r < \min(2s(f), n)$, there exists a
center $x_0$ such that $(f|_{\cB(x_0,r)})^\Mj \neq f$.  
Let $x$ be a point with sensitivity $s(f)$, and let $S \subset [n]$ be
the set of $s(f)$ sensitive coordinates at $x$. We will pick
$x_0$ so that $d(x,x_0) = r+1$ as follows.
If $r +1 \leq s(f)$, we obtain $x_0$ from $x$ by flipping some $r +1$
coordinates from $S$. If $r +1 > s(f)$, then we obtained $x_0$ from $f$ by flipping all the
coordinates in $S$, and any $r + 1 -s(f)$ other coordinates
$T \subseteq [n]\setminus S$. The condition $r +1 \leq n$ guarantees
that a subset of the desired size exists, while $r + 1 \leq 2s(f)$
enures that $|T| \leq |S|$. 

Since $d(x,x_0) = r+1$, the value at $x$ is inferred using the
Majority rule applied to the neighbors of $x$ in $\cB(x_0,r)$. These
neighbors are obtained by either flipping coordinates in $S$ or
$T$ (where $T$ might be empty). The former disagree with $f(x)$ while the
latter agree. Since $|S| \geq |T|$, the Majority rule either labels
$x$ wrongly, or leaves it undetermined (in the case when $r =
2s(f)$). This shows that $(f|_{\cB(x_0,r)})^\Mj \neq f(x)$, hence
$r^\Mj \geq \min(2s(f),n)$.
\end{Proof}

In contrast with Lemma \ref{lem:par}, quantifying over all
centers $x_0$ in the defintion of $r^\Par$ is crucial for the lower
bound in Lemma \ref{lem:maj}. This is seen by  considering the $n$-variable
$\mathrm{OR}$ function, where the sensitivity is $n$. Applying the Majority rule to a ball of radius $2$ around $0^n$
returns the right function, but if we center the ball at $1^n$, then
the Majority rule cannot correctly infer the value at $0^n$, so it
needs to be part of the advice, hence $r^\Mj(\mathrm{OR}) =n$.

\subsection{Agreement of the Majority and Parity Rule}

Lemmas \ref{lem:par} and \ref{lem:maj} can be viewed as
alternate characterizations of the degree and sensitivity of a Boolean
function. The degree versus sensitivity conjecture asserts that both
these rules work well (meaning that they only require the values on a
small ball as advice) for the same class of functions. Given that the
rules are so simple, and seem so different from each other, we find this assertion surprising.

In particular,  Conjecture \ref{conj:2} is equivalent to the following
statement:

\begin{Conj}
\label{conj:3}
There exists constants $d_1, d_2$ such that
\begin{align} 
\label{eq:r-eq1}
r^\Par(f) \leq d_1(r^\Mj)^{d_2}.
\end{align}
\end{Conj}

Along similar lines, one can use Theorem
\ref{thm:ns}, due to \cite{NisanSzegedy:94}, to show that the Majority
rule recovers low-degree Boolean functions:
\begin{align} 
\label{eq:r-eq2}
r^\Mj(f) \leq 8(r^\Par(f))^2.
\end{align}
Their proof uses Markov's inequality from analysis. It might be
interesting to find a different proof, which one could hope to 
extend to proving Equation \eqref{eq:r-eq1} as well.


\section{Conclusion (and more open problems)}

We have presented the first upper bounds on the computational complexity of low
sensitivity functions. We believe this might be a promising alternative
approach to Conjecture \ref{conj:1} as opposed to getting improved
bounds on specific low level measures like block sensitivity or
decision tree depth \cite{KK,AmbainisBGMSZ:14,AmbainisS:11}. 

 Conjecture \ref{conj:1} implies much stronger upper bounds than are given by
 our results.  We  list some of the ones
which might be more approachable given our results:

\begin{enumerate}
\item Every sensitivity $s$ function has a $\mathsf{TC}_0$ circuit of size
$n^{O(s)}$. 
\item Every sensitivity $s$ function has a polynomial threshold function
  (PTF) of degree $\poly(s)$.
\item Every sensitivity $s$ function $f$ has $\deg_2(f) \leq s^c$ for some
  constant $c$.
\end{enumerate}

\bibliographystyle{alpha}
\bibliography{allrefs}

\newcommand{\etalchar}[1]{$^{#1}$}
\begin{thebibliography}{ABG{\etalchar{+}}14}

\bibitem[ABG{\etalchar{+}}14]{AmbainisBGMSZ:14}
Andris Ambainis, Mohammad Bavarian, Yihan Gao, Jieming Mao, Xiaoming Sun, and
  Song Zuo.
\newblock Tighter relations between sensitivity and other complexity measures.
\newblock In {\em Automata, Languages, and Programming - 41st International
  Colloquium, {ICALP} 2014}, pages 101--113, 2014.

\bibitem[AP14]{AmbainisP:14}
Andris Ambainis and Krisjanis Prusis.
\newblock A tight lower bound on certificate complexity in terms of block
  sensitivity and sensitivity.
\newblock In {\em MFCS}, pages 33--44, 2014.

\bibitem[APV15]{AmbainisPV:15}
Andris Ambainis, Krisjanis Prusis, and Jevgenijs Vihrovs.
\newblock Sensitivity versus certificate complexity of boolean functions.
\newblock {\em CoRR}, abs/1503.07691, 2015.

\bibitem[AS11]{AmbainisS:11}
Andris Ambainis and Xiaoming Sun.
\newblock New separation between {\textdollar}s(f){\textdollar} and
  {\textdollar}bs(f){\textdollar}.
\newblock {\em CoRR}, abs/1108.3494, 2011.

\bibitem[AV15]{AmbainisV:15}
Andris Ambainis and Jevgenijs Vihrovs.
\newblock Size of sets with small sensitivity: a generalization of simon's
  lemma.
\newblock In {\em Theory and Applications of Models of Computation - 12th
  Annual Conference, {TAMC} 2015}, pages 122--133, 2015.

\bibitem[BdW02]{BuhrmandeWolf:02}
H.~Buhrman and R.~de~Wolf.
\newblock Complexity measures and decision tree complexity: a survey.
\newblock {\em Theoretical Computer Science}, 288(1):21--43, 2002.

\bibitem[CDR86]{CookDR:86}
Stephen~A. Cook, Cynthia Dwork, and R{\"{u}}diger Reischuk.
\newblock Upper and lower time bounds for parallel random access machines
  without simultaneous writes.
\newblock {\em {SIAM} J. Comput.}, 15(1):87--97, 1986.

\bibitem[GKS15]{GKS:15}
Justin Gilmer, Michal Kouck{\'{y}}, and Michael~E. Saks.
\newblock A new approach to the sensitivity conjecture.
\newblock In {\em Proceedings of the 2015 Conference on Innovations in
  Theoretical Computer Science, {ITCS} 2015}, pages 247--254, 2015.

\bibitem[HKP11]{Chicago}
Pooya Hatami, Raghav Kulkarni, and Denis Pankratov.
\newblock {\em Variations on the Sensitivity Conjecture}.
\newblock Number~4 in Graduate Surveys. Theory of Computing Library, 2011.

\bibitem[Jac30]{jackson1930}
Dunham Jackson.
\newblock The theory of approximation.
\newblock {\em New York}, 19:30, 1930.

\bibitem[Juk12]{Jukna:12}
S.~Jukna.
\newblock {\em Boolean Function Complexity: Advances and Frontiers}.
\newblock Springer, 2012.

\bibitem[KK04]{KK}
Claire Kenyon and Samuel Kutin.
\newblock Sensitivity, block sensitivity, and l-block sensitivity of {B}oolean
  functions.
\newblock {\em Information and Computation}, pages 43--53, 2004.

\bibitem[MORS07]{MOR+:07eccc}
K.~Matulef, R.~O'Donnell, R.~Rubinfeld, and R.~Servedio.
\newblock {Testing Halfspaces}.
\newblock Technical Report 128, Electronic Colloquium in Computational
  Complexity, 2007.
\newblock Full version in FOCS 2007.

\bibitem[Nis91]{Nisan:91}
Noam Nisan.
\newblock Crew prams and decision trees.
\newblock {\em SIAM Journal on Computing}, 20(6):999--1007, 1991.

\bibitem[NS64]{NewSha1964}
DJ~Newman and HS~Shapiro.
\newblock Jackson's theorem in higher dimensions.
\newblock In {\em On Approximation Theory/{\"U}ber Approximationstheorie},
  pages 208--219. Springer, 1964.

\bibitem[NS94]{NisanSzegedy:94}
N.~Nisan and M.~Szegedy.
\newblock {On the degree of {B}oolean functions as real polynomials}.
\newblock {\em Comput. Complexity}, 4:301--313, 1994.

\bibitem[O'D14]{Ryan}
Ryan O'Donnell.
\newblock {\em Analysis of {B}oolean Functions}.
\newblock Cambridge University Press, 2014.

\bibitem[Rub95]{Rubinstein}
David Rubinstein.
\newblock Sensitivity vs. block sensitivity of {B}oolean functions.
\newblock {\em Combinatorica}, pages 297--299, 1995.

\bibitem[Sim82]{Simon82}
Hans{-}Ulrich Simon.
\newblock A tight omega(log log n)-bound on the time for parallel ram's to
  compute nondegenerated boolean functions.
\newblock {\em Information and Control}, 55(1-3):102--106, 1982.

\bibitem[Wei85]{weierstrass1885}
Karl Weierstrass.
\newblock {\"U}ber die analytische darstellbarkeit sogenannter
  willk{\"u}rlicher functionen einer reellen ver{\"a}nderlichen.
\newblock {\em Sitzungsberichte der K{\"o}niglich Preu{\ss}ischen Akademie der
  Wissenschaften zu Berlin}, 2:633--639, 1885.

\end{thebibliography}

\appendix

\section{Omitted Proofs and Results}

\noindent {\bf Proof of Lemma \ref{lem:bias}:}
It suffices to prove the bound for $s_1$. 
The proof is by induction on the dimension $n$. Observe that if
$\mu_1(f) =1$ then the claim is trivial, so we may assume $\mu_1 \in (0,1)$.
In the base case $n =1$, we must have $\mu_1 = 1/2$, in which case
$s(f) =1$ so the claim holds. 

For any $i \in [n]$, let $f_{i,1} = f|_{x_i =1}$ and $f_{i,0} = f|_{x_i =0}$ denote the
restrictions of $f$ to the subcubes defined by $x_i$. These are each
functions on variables in $[n]\setminus \{i\}$. Then
\[ \mu_1(f) = \frac{\mu_1(f_{i,1}) + \mu_1(f_{i,0})}{2}. \]
 
If there exists $b \in \zo$ such that $0 < \mu_1(f_{i,b}) \leq \mu_1(f)$
then we can apply the inductive claim to the restricted function
$f_{i,b}$ to conclude that there exists a point $x \in f^{-1}_{i,b}(1)$ so that 
\[ s(f,x) \geq \log\left(\frac{1}{\mu_1(f_{i,b})}\right) \geq
\log\left(\frac{1}{\mu_1(f)}\right). \]

If not, it must be that
$f(x) =1$ implies $x_i =b$ for some $b \in \zo$, so that 
\[\mu_1(f_{i,b}) = 2\mu_1(f) \ \text{and} \ \mu_1(f_{i,1-b}) = 0.\]
But then every point $x \in f^{-1}(1)$ is sensitive to $x_i$. 
Further, we can apply the inductive hypothesis to $f_{i,b}$, to conclude that there
exists $x \in f_{i,b}^{-1}(1)$ such that $x$ is sensitive to 
\[ \log\left(\fr{\mu_1(f_{i,b})}\right) = \log\left(\fr{2\mu_1(f)}\right) =  \log\left(\fr{\mu_1(f)}\right) -1 \]
coordinates from $[n]\setminus \{i\}$. Since $x$ is also sensitive to
$i$, we have
\[ s_1(f,x) \geq \log\left(\frac{1}{\mu_1(f)}\right).\]

For the final claim, assume the above bound holds with equality. Then there do not
exist $i \in [n], b \in \zo$ such that $0 < \mu_1(f_{i,b}) <
\mu_1(f)$ (if they did exist then we would get a stronger bound). So for every $i$,
either $\mu_1(f_{i,b}) = 0$ for some $b$, or $\mu_1(f_{i,0}) = \mu_1(f_{i,1})$. In the former case, the set
$f^{-1}(1)$ is contained in the subcube $x_i = b$. In the latter case,
by induction we may assume that $f^{-1}(1)$ restricted to both  $x_i
=0$ and $x_1 =1$ is a subcube of density exactly $\mu_1(f)$ in
$\zo^{[n] \setminus \{i\}}$, so every point in these subcubes must have sensitivity $\log(1/\mu_1(f))$ 
We further claim that the two subcubes are identical as functions on
$\zo^{[n]\setminus \{i\}}$. If they were not identical, then some
point (in each subcube) would be sensitive to coordinate $i$, but then
this point would have sensitivity at least $\log(1/\mu_1(f)) +1$.

This implies that $f^{-1}(1)$ is a subcube defined by the equations
$x_i = 1- b$ for all pairs $(i,b)$ such that $\mu_1(f_{i,b}) = 0$.
\qed

\subsection{A Top-Down Algorithm}

Next we describe a ``top-down'' algorithm for computing $f(x)$ where $f$ is a function of sensitivity $s$.  This algorithm has a similar performance bound
to our ``bottom-up'' algorithm described earlier.

Associate the bit string $x \in \zo^n$ with the integer $z(x) =
\sum_{i=1}^nx_i2^i$, and let $x < x'$ if $z(x) < z(x')$. We refer to
this as the \rl\ ordering on strings. 
 
The top-down algorithm also takes the values of $f$ on $\cB(0^n,2s)$
as advice. Given an input $x \in \zo^n$ where we wish to evaluate $f$,
we recursively evaluate $f$ at the first  $2s +1$
neighbors of $x$ of Hamming weight $\wt(x) -1$ in the \rl\ order. The
recursion bottoms out when we reach an input of weight $2s$. 
The restriction to small elements in the \rl\ order ensures that
the entire set of inputs on which we need to evaluate $f$ is small.  A detailed description of the algorithm follows:

\medskip

\myalgo{Top-Down}{
{\bf Advice: }$f$ at all points in $\cB(0^n,2s)$.\\
{\bf Input: }$x \in \zo^n$.

\begin{enumerate}
\item If $\wt(x) \leq 2s$ or if $f(x)$ has been computed before, return $f(x)$.
\item Otherwise, let $x_1,\ldots,x_{2s +1}$ be the $2s+1$ smallest
  elements in $N(x)$ of weight $\wt(x) -1$ in the \rl\ order. If some 
  $f(x_i)$ has not been computed yet, compute it recursively and store
  the value.
\item Return $f(x) = \Maj{i \in [2s+1]}{f(x_i)}$.
\end{enumerate}
}

\medskip

The key to the analysis is the following lemma.

\begin{Lem}
\label{lem:revlex}
Let $\wt(x) =d$. For $2s \leq k \leq d$, the number of weight $k$
vectors $x'$ for which $f(x')$ is computed by the \emph{{\bf Top-Down}} algorithm is bounded
by
\[ {d - k + 2s \choose d -k} \leq d^{2s}.\]

\end{Lem}
\begin{Proof}
Given $x \in \zo^n$, for $t \leq \wt(x)$, let $R(t) \subseteq [n]$ denote
the $t$ largest indices $i \in [n]$ where $x_i =1$. 
We claim that all vectors $x'$ with $\wt(x') = k$ that are generated by the algorithm 
are obtained by setting $d-k$ indices in $R(d - k +2s)$ to $0$. This
claim clearly implies the desired bound.

The claim is proved by induction on $d- k$. The case $d -k  =1$
is easy to see, since the $2s +1$ smallest neighbors of $x$ in the
colex order (of weight one less than $x$) are each obtained by setting one of the
indices in $R(2s +1)$ to $0$.
For the inductive case, assume that $\wt(y) =k$, and that $y$ is
generated as a neighbor of $y'$ with $\wt(y') = k +1$. Inductively,   $y'$ is obtained from $x$ by setting 
indices in $S \subset R(d -k -1 +2s)$ to $0$, where $|S| = d - k -1$, and hence
leaving $2s$ of them $1$. Thus the $2s$ smallest neighbors of $y'$ are
obtained by setting indices in $R(d -k -1 +2s)\setminus S$ to $0$, and the $(2s +1)$th
neighbor is obtained by setting the $(d - k +2s)$th $1$ from the right
to $0$. In both cases, we get $d - k$ indices from $R(d -k +2s)$ being
set to $0$. This completes the induction.
\end{Proof}

\begin{Thm}
\label{thm:algQ}
The \emph{{\bf Top-Down}} algorithm computes $f(x)$ for any input $x$ in time $O(sn^{2s+1})$ using space
$O(n^{2s+1})$.
\end{Thm}
\begin{Proof}
By the preceding lemma,
for an input $x$ of weight $d$, the total number of $x'$ for which $f(x')$ is
computed and stored is bounded by
\[ \sum_{k=2s}^dd^{2s} \leq d^{2s +1} \leq n^{2s+1}.\]
The cost of computing $f$ at $x$ given $f$'s values at the relevant $2s+1$ neighbors of $x$ (see Step~3) is $O(s)$, so on average the amortized cost for computing $f(x)$ at each $x$ is bounded by $O(s)$.  Hence overall the running time and space are bounded by $O(sn^{2s +1})$ and
$O(n^{2s +1})$ respectively.
\end{Proof}

\end{document}